\documentclass[11pt]{article}

\usepackage{url}
\usepackage[binary-units=true]{siunitx}
\usepackage{hyperref}
\usepackage{graphicx}
\usepackage[labelformat=simple]{subcaption}

\usepackage{algorithm, algorithmicx}
\usepackage{algpseudocode}
\usepackage{color}
\usepackage{booktabs}
\usepackage{bookmark}
\usepackage{epstopdf}
\usepackage{cite}
\usepackage{fullpage}
\usepackage{amsthm} 

\algdef{SE}[EVENT]{Event}{EndEvent}[1]{\textbf{upon event}\ #1\ \algorithmicdo}{\algorithmicend\ \textbf{event}}%
\algtext*{EndEvent}

\algdef{SE}[PHASE]{Phase}{EndPhase}[1]{\textbf{Phase}\ #1\ \algorithmicdo}{\algorithmicend\}}%
\algtext*{EndPhase}

\algdef{SE}[PROCESS]{Process}{EndProcess}[1]{\textbf{Process}\ #1\ \algorithmicdo}{\algorithmicend\}}%
\algtext*{EndProcess}

\algdef{SE}[THREAD]{Thread}{EndThread}[1]{\textbf{Thread}\ #1\ \algorithmicdo}{\algorithmicend\}}%
\algtext*{EndThread}

\newtheorem{lemma}{Lemma}
\newtheorem{theorem}{Theorem}
\newtheorem{axiom}{Axiom}
\newcommand{\sol}{MANA}

\newcommand{\FINAL}{no}
\ifthenelse{\equal{\FINAL}{yes}}{%
\newcommand{\finalVersion}[1]{{\color{red} {\bf FINAL VERSION:} #1}}
} {
\newcommand{\finalVersion}[1]{}
}

\title{MANA for MPI: MPI-Agnostic Network-Agnostic Transparent Checkpointing}
\date{}

\author{Rohan Garg \\
  Northeastern Univ. \\
  Boston, MA, USA \\
  \url{rohgarg@ccs.neu.edu}
\and
  Gregory Price  \\
  Raytheon Company \\
  Annapolis Junction, MD, USA \\
  \url{gregory.m.price@raytheon.com}
\and
  Gene Cooperman \\
  Northeastern Univ. \\
  Boston, MA, USA \\
  \url{gene@ccs.neu.edu}
}

\begin{document}
\maketitle

\begin{abstract}
Transparently checkpointing MPI for fault tolerance and load balancing is
a long-standing problem in HPC.  The problem has been complicated by the
need to provide checkpoint-restart services for all combinations of an
MPI implementation over all network interconnects.  This work presents
\sol{} (MPI-Agnostic Network-Agnostic transparent checkpointing), a
single code base which supports all MPI implementation and interconnect
combinations.  The agnostic properties imply that one can checkpoint
an MPI application under one MPI implementation and perhaps over TCP,
and then restart under a second MPI implementation over InfiniBand on
a cluster with a different number of CPU cores per node.  This technique
is based on a novel {\em split-process} approach, which enables two separate
programs to
co-exist within a single process with a single address space.
This work overcomes the limitations of the two most widely adopted
transparent checkpointing solutions, BLCR and DMTCP/InfiniBand,
which require separate modifications to each MPI implementation and/or
underlying network API.  The runtime overhead is found to be
insignificant both for checkpoint-restart within a single host, and
when comparing a local MPI computation that was migrated to a remote
cluster against an ordinary MPI computation running natively on that
same remote cluster.
\end{abstract}

\section{Introduction}
\label{sec:intro}
The use of transparent or system-level checkpointing for MPI is
facing a crisis today.  The most common transparent checkpointing packages
for MPI in recent history are either declining in usage, or abandoned
entirely.  These checkpointing packages include:
the Open~MPI~\cite{OpenMPICheckpoint07} checkpoint-restart service,
the MVAPICH2~\cite{GaoEtAl06} checkpoint-restart service,
DMTCP for MPI~\cite{ansel2009dmtcp}, MPICH-V~\cite{Bouteiller06}, and 
a fault-tolerant BLCR-based ``backplane'', CIFTS~\cite{GuptaEtAl09}.
We argue existing transparent or system-level checkpoint approaches
share common issues that makes long-term maintenance impractical.
In particular, the HPC community requires checkpoint-restart support
for any of $m$~popular MPI implementations over $n$~different
network interconnects.

We propose MPI-Agnostic Network-Agnostic transparent checkpointing
(\sol{}), a single code base that can support all combinations
of the many MPI implementations and network libraries that are in wide use.
In particular, it supports all $m\times n$ combinations, where
$m$ is the number of MPI implementations and
$n$ is the number of underlying network libraries.
The new approach, based on a {\em split-process}, is fully transparent
to the underlying MPI, network, libc library, and underlying Linux kernel.
\sol{} is free and open-source software~\cite{manaGithub}.  (Transparent checkpointing supports
standard system-level checkpointing, but it can alternatively be customized in
an application-specific manner.)

We begin by distinguishing this work from that of Hursey
et~al.~\cite{hursey2009interconnect}, which demonstrated a
network-agnostic implementation of checkpointing for a single MPI
implementation (for Open~MPI).  Hursey's work adds network-agnostic
checkpointing by ``taking down'' the network during checkpoint and
``building it up'' upon resuming --- but only within the single Open~MPI
implementation.  Further, it requires maintenance of $n$~code bases,
where $n$ is the number of network libraries that Open~MPI supports.

In contrast, The new approach of \sol{} employs a single code base
that exists {\em external to any particular MPI implementation}.
At checkpoint time, \sol{} ``disconnects'' the application from the MPI
and network libraries, and at the time of resuming, it ``reconnects''
to the MPI and network libraries.  Further, at the time of restart,
\sol{} starts a new and independent MPI session (even possibly using
a newer version of the original MPI implementation (perhaps due to a
system upgrade), or it may even switch MPI implementations.

Next, we present three case studies to demonstrate the declining usage of
transparent checkpointing, and how maintenance costs have factored into
supressing adoption of even semi-agnostic systems.

First, we consider Hursey et al.~\cite{hursey2009interconnect}
and Open-MPI.  Open~MPI developers created a novel and elegant
checkpoint-restart service that was {\em network-agnostic}, with the
ability to checkpoint under network~A and restart under network~B.
The mechanism presented by Hursey et al warrants careful analysis
of how it provides {\em network-agnostic} checkpointing (a primary
goal of \sol{}).  Their implementation lifts checkpoint code out of
interconnect drivers, applying a Chandy/Lamport~\cite{ChandyLamport}
checkpoint algorithm to an abstraction of ``MPI Messages''.

However, their choice to implement within the MPI library may have
ultimately suppressed widespread adoption.  Hursey et al\hbox{.} can
support multiple MPI implementations only by requiring
each MPI implementation to individually integrate the Hursey approach
and maintain support for taking down and restoring network
interconnects for each supported network library.
This imposes a significant maintenance penalty on packages maintainers.
Even Open-MPI support remains in question.  As of this writing in 2019,
the Open-MPI FAQ says: ``Note: The checkpoint/restart support was last
released as part of the v1.6 series. $\ldots$ This feature is looking
for a maintainer.''~\cite{openmpiCrService}.

The second case study concerns BLCR.  MPICH and several other MPI
implementations adopted BLCR for checkpointing. BLCR is based on a kernel
module that checkpoints the local MPI rank.  BLCR lacks System~V shared memory
support (widely used for intra-node communication), which severely limits its
use in practice.  As of this writing, BLCR~0.8.5 (appearing in~2013) was
the last officially supported version~\cite{blcrDownloads}, and formal
testing of the BLCR kernel module stopped with Linux~3.7
(Dec.,~2012)~\cite{blcrKernelTesting}.  Here, again, we argue that BLCR
declined not due to any fault with BLCR, but due to the difficulty of
supporting or maintaining common interconnects.

The third case study concerns DMTCP.  As discussed above,
DMTCP/InfiniBand is {\em MPI-agnostic}, but not {\em network-agnostic},
requiring a plugin for each network interconnect.  While it supports InfiniBand
~\cite{cao2014transparent}, and partially supports Intel
Omni-Path~~\cite[Chapter~6]{caophdthesis}, DMTCP does not support Cray GNI Aries
network, the Mellanox extensions to InfiniBand (UCX and FCA), the libfabric
API~\cite{libfabric}, and many others.

Separate {\em MPI-agnostic} or {\em network-agnostic} checkpoint systems have
not been widely adopted, it seems, in-part due to the maintenance costs.
\sol{} attempts to resolve this.

The {\em split-process} approach eliminates this
maintenance penalty, supporting checkpoint/restart on all
combinations of MPI library and interconnect with a single codebase.  In a
{\em split-process}, a single system process contains two programs in its
memory address space.  The two programs are an MPI proxy application
(denoted the lower-half) and the original MPI application code (denoted
the upper-half).  \sol{} tracks which memory regions belong to the upper and
lower halves.  At checkpoint time, only upper-half memory regions are saved.
\sol{} is also fully transparent to the specific MPI implementation, network,
libc library and Linux kernel.

At restart time, \sol{} initializes a new MPI library and underlying
interconnect network in the lower half of a process.  The checkpointed
MPI application code and data is then copied in and restored into the upper half
from the checkpoint image file.
By initializing a new MPI library at the time of restart, \sol{}
provides excellent load-balancing support without the need for additional
logic.  The fresh initialization inherently detects the correct number of
CPU cores per node, optimizes the topology as MPI ranks from the same node
may now be split among distinct nodes (or vice verse), re-optimizes the
rank-to-host bindings for any MPI topology declarations in the MPI application,
and so on. (See Section~\ref{sec:futureWork} for further discussion.)

\sol{} maintains low runtime overhead, avoiding RPC by taking advantage
of a split-process to make library calls within the same process.  Other
Proxy-based approaches had been previously used for checkpointing in general
applications~\cite{zandy1999process} and CUDA (GPU-accelerated)
applications~\cite{garg2018crum,gtc2016crcuda,takizawa2011checl}.
However, such approaches incur significant overhead due to context switching
and copying buffers between the MPI application and the MPI proxy.

Scalability is a key criterion for \sol{}.  This criterion motivates
the use of version~3.0 of DMTCP~\cite{ansel2009dmtcp} as the
underlying package for transparent checkpointing.  DMTCP was previously
used to demonstrate petascale-level transparent checkpointing using Lustre.
In particular this was applied to MPI-based HPCG over 32,752 CPU
cores (checkpointing an aggregate 38~TB in 11~minutes), and to
MPI-based NAMD over 16,368 cores (checkpointing an aggregate 10~TB in
2.6~minutes)~\cite{cao2016system}.  DMTCP employs a stateless
centralized checkpoint coordinator for its results, and \sol{} re-uses
for its own purposes this same coordinator.  (The single coordinator
is not a barrier to scalability, since it can use a broadcast
tree for communication with its peers.)

Next in this work, Section~\ref{sec:design} describes the design issues in
fitting the split-process concept to checkpoint MPI.  In particular, the
issue of checkpointing during MPI collective communications is discussed.
Section~\ref{sec:evaluation} presents an experimental evaluation.
Section~\ref{sec:futureWork} discusses current and future inquiries opened
up by these new ideas.
Section~\ref{sec:relatedWork} discusses related work.
Finally, Section~\ref{sec:conclusion} presents the conclusion.

\section{\sol{}: Design and Implementation}
\label{sec:design}

Multiple aspects of the design of \sol{} are covered in this section.
Section~\ref{sec:splitProcesses} discusses the design for supporting
a split-process.
Section~\ref{sec:mpiIds} discusses the need to save and restore
persistent MPI opaque objects, such as communicators, groups and topologies.
Section~\ref{sec:pointToPoint} briefly discusses the commonly used
algorithm to drain point-to-point MPI messages in transit prior to intiaiting a
checkpoint.
Sections~\ref{sec:mpiCollectivesOverview} and~\ref{sec:mpiCollectivesAlgo}
present a new two-phase algorithm (Algorithm~\ref{algo:mpi2pcSimple}),
which enables checkpointing in-progress MPI collective communication
calls in a fully agnostic environment.
Finally, Sections~\ref{sec:implementation} and~\ref{sec:crpackage}
present details of the overall implementation of \sol{}.

\subsection{Upper and Lower Half: Checkpointing with an Ephemeral MPI Library}
\label{sec:splitProcesses}

In this section, we define the {\em lower half} of a split-process as
the memory associated with the MPI library and dependencies,
including network libraries.  The {\em upper half} is the remaining
Linux process memory associated with the MPI application's code, data,
stack, and other regions (e.g., environment variables).  The terms
{\em lower half} and {\em upper half} are in analogy with the upper
and lower half of a device driver in an operating system kernel.
This separation into lower and upper half does not involve additional
threads or processes. Instead, it serves primarily to tag memory so that only
upper half memory will be saved or restored during checkpoint and restart.
Section~\ref{sec:implementation} describes an additional ``helper
thread'', but that thread is active only during checkpoint and restart.

Libc and other system libraries may appear in both the lower half as a
dependency of the MPI libraries, and the upper half as an independent
dependency of the MPI application.

This split-process approach allows \sol{} to balance two
conflicting objectives: a shared address space; and isolation of
upper and lower halves.  The isolation allows \sol{} to omit the lower
half memory (an ``ephemeral'' MPI library)
when it creates a checkpoint image file.  The shared
address space allows the flow of control to pass efficiently from
the upper-half MPI application to the lower-half MPI library through
standard C/Fortran calling conventions, including call by reference.
As previously noted, Remote Produce Calls (RPC) are not employed.

{\em Isolation} is needed so that at checkpoint time,
the lower half can be omitted from the checkpoint image, and
at the time of restart, replaced with a small ``bootstrap'' MPI program
with new MPI libraries.  The bootstrap program calls {\tt MPI\_Init()} and
each MPI process discovers its MPI rank via a call to {\tt MPI\_Rank()}.
The memory present at this time becomes the lower half.
The MPI process then restores the upper-half memory from a checkpoint
image file corresponding to the MPI rank id.  Control is then
transferred back to the upper-half MPI application, and the stack
in the lower half is never used again.

{\em Shared address space} is needed for efficiency. A dual-process
proxy approach was explored in~\cite[Section~IV.B]{garg2018crum} and
in~\cite[Section~IV.A]{takizawa2011checl}. The former work reported
a 6\%~runtime overhead for real-world CUDA applications, and the
latter work reported runtime overheads in excess of 20\% for some
OpenCL examples from the NVIDIA SDK~3.0.  In contrast, Section~\ref{sec:evaluation}
reports runtime overheads less than~2\% for \sol{} under older Linux
kernels, and less than~1\% runtime overhead for recent Linux kernels.

Discarding the lower half greatly simplifies the task of checkpointing. By
discarding the lower half, the MPI application in the upper half
appears as an isolated process with no inter-process communication.
Therefore, a single-process checkpointing package can create a
checkpoint image.

A minor inconvenience of this split-process approach is that calls to
{\tt sbrk()} will cause the kernel to extend the process heap in the data
segment. Calls to {\tt sbrk()} can be caused by invocations of {\tt malloc()}.
Since the kernel has no concept of a split-process, the kernel may choose, for
example, to extend the lower half data segment after restart since that
corresponds to the original program seen by the kernel before the upper-half
memory is restored.  \sol{} resolves this by interposing on calls to
{\tt sbrk()} in the upper-half libc, and then inserts calls to {\tt mmap()}
to extend the heap of the upper-half.

Finally, \sol{} employs coordinated checkpointing, and a checkpoint
coordinator sends messages to each MPI rank at the time of checkpoint (see
Sections~\ref{sec:pointToPoint},~\ref{sec:mpiCollectivesOverview}
and~\ref{sec:mpiCollectivesAlgo}).  MPI opaque objects (communicators,
groups, topologies) are detected on creation and restored on restart
(see Section~\ref{sec:mpiIds}).  This is part of a broader strategy by which
MPI calls with persistent effects (such as creation of these opaque objects)
are recorded during runtime and replayed on restart.

\subsection{Checkpointing MPI Communicators, Groups, and Topologies}
\label{sec:mpiIds}

An MPI application can create communication subgroups and topologies to
group processes for ease of programmability and efficient communication.
MPI implementations provide opaque handles to the application
as a reference to a communicator object or group.

\sol{} interposes on all calls that refer to these opaque identifiers,
and virtualizes the identifiers. At runtime, \sol{} records any MPI calls
that can modify the MPI communication state, such as \texttt{MPI\_Comm\_create},
\texttt{MPI\_Group\_incl}, etc.  On restart, \sol{} recreates the MPI
communicator state by replaying the MPI calls using a new MPI library.  The
runtime virtualization of identifiers allows the application to continue
running with consistent handles across checkpoint-restart.

A similar checkpointing strategy also works for other opaque identifiers, such
as, MPI derived datatypes, etc.

\subsection{Checkpointing MPI Point-to-Point Communication}
\label{sec:pointToPoint}

Capturing the state of MPI processes requires quiescing the process threads,
and preserving the process memory to a file on the disk. However, this alone
is not sufficient to capture a consistent state of the computation. Any MPI
messages sent but not yet received at the time of quiescing processes must
also be saved as part of the checkpoint image.

\sol{} employs a variation of an all-to-all bookmark exchange algorithm
to reach this consistent state.  LAM/MPI~\cite{CheckpointLAMMPI05b} demonstrated
the efficacy of a such a Chandy/Lamport~\cite{ChandyLamport} algorithm for
checkpointing MPI applications. Hursey et al.~\cite{hursey2009interconnect}
lifted this mechanism out of interconnect drivers and into the MPI library.
\sol{} further lifts this mechanism outside the MPI library, and into a
virtualized MPI API.

An early prototype of \sol{} demonstrated a na\"ive application of this bookmark
exchange algorithm was sufficient for handling pre-checkpoint draining for
point-to-point communication; however, collective-communication calls may have
MPI implementation effects that can determine when it is ``safe'' to begin
a checkpoint. For this reason, a na\"ive application to the entire API was
insufficient to ensure correctness.  This is discussed in
Section~\ref{sec:mpiCollectivesOverview}.

\subsection{Checkpointing MPI Collectives: Overview}
\label{sec:mpiCollectivesOverview}

The MPI collective communications primitive involves communication amongst all
or a program-defined subset of MPI ranks (as specified by the MPI communicator
argument to the function). The internal behavior of collectives are specific
to each MPI implementation, and so it is not possible to make guarantees about
their behavior, such as when and how messages are exchanged when ranks are
waiting for one or more ranks to enter the collective.

In prior work ~\cite{CheckpointLAMMPI05b,hursey2009interconnect}, internal
knowledge of the MPI library state was required to ensure that checkpointing
would occur at a ``safe'' state.  In particular, Hursey et al.
~\cite{hursey2009interconnect} required interconnect drivers be classified
as ``checkpoint-friendly'' or ``checkpoint-unfriendly'', changing behavior based
on this classification.  As \sol{} lives outside the MPI library,
a naive application of the Hursey et al. algorithm can have effects that cross
the upper and lower half boundaries of an MPI rank (for example, when shared
memory is being used for MPI communication).

This problem occurs because of the truly {\em network-agnostic} trait of \sol{}.
As \sol{} has no concept of transport level constructs, it cannot determine
what ``safe'' means in context of collectives. To correct this, \sol{}'s
support for collective communication requires it to maintain the following
invariant:
\begin{quote}
  {\em No checkpoint must take place while a rank is inside a
  collective communication call.}
\end{quote}

There exists one exception to this rule: a {\em trivial barrier}.
A {\em trivial barrier} is a simple call to MPI\_Barrier().  This call produces
no side effects on an MPI rank, and so it can be safely interrupted during
checkpoint, and then re-issued when restarting the MPI application.  This is
possible due to the split-process architecture of \sol{}, as
{\em trivial barrier} calls occur exclusively in the lower half, which
is discarded and replaced across checkpoint and restart.  \sol{}
leverages this exception to build a solution for all other collective calls.

As we discuss \sol{}'s algorithm for checkpointing collective calls,
we take into consideration three subtle, but important, concerns.
\begin{description}
\item[Challenge I (consistency):]
  In the case of a single MPI collective communication call, there is a
  danger that rank~A will see a request to checkpoint before entering
  the collective call, while rank~B will see the same request after
  entering the collective call, in violation
  of \sol{}'s invariant.  Both ranks might report that they are ready
  to checkpoint, and the resulting inconsistent snapshot will create problems
  during restart.  This situation could arise, for example, if the
  message from the checkpoint coordinator to rank~B is excessively
  delayed in the network.  To resolve this, \sol{} introduces
  a two-pass protocol in which the coordinator makes a request (sends
  an intend-to-checkpoint message), each MPI rank acknowledges with its
  current state, and finally the coordinator posts a checkpoint request
  (possibly preceded by extra iterations).
\item[Challenge II (progress and latency):]
  Given the aforementioned solution for consistency, long delays may occur
  before a checkpoint request can be serviced.
  It may be that rank~A has entered the
  barrier, and rank~B will require several hours to finish a
  task before entering the barrier.  Hence, the two-pass protocol may
  create unacceptable delays before a checkpoint can be taken.
  Algorithm~\ref{algo:mpi2pcSimple} addresses this by introducing a
  {\em trivial barrier} prior to the collective communication call.
  We refer to this as a {\em two-phase algorithm} since each collective call is
  now replaced by a wrapper function that invokes a trivial barrier call
  (phase~1) followed by the original collective call (phase~2).
\item[Challenge III (multiple collective calls):]
  Until now, it was assumed that at most one MPI collective communication
  call was in progress at the time of checkpoint.  It may happen that
  there are multiple ongoing collective calls.  During the time that
  some MPI ranks exit from a collective call, it may happen that MPI
  ranks associated with an independent collective call have left the MPI
  {\em trivial barrier} (phase~1) and have now entered the real collective call
  (phase~2).  As a result, servicing a checkpoint may be excessive delayed.
  To solve this, we introduce an intend-to-checkpoint message, such that no
  ranks will be allowed to enter phase~2, and extra iterations will be inserted
  into the request-acknowledge protocol between coordinator and MPI rank.
\end{description}

\hbox{\ \ }

\subsection{Checkpointing MPI Collectives: Detailed Algorithm}
\label{sec:mpiCollectivesAlgo}

Here we present a single algorithm (Algorithm~\ref{algo:mpi2pcSimple}) for
checkpointing MPI collectives which contains the elements described in
Section~\ref{sec:mpiCollectivesOverview}: a multi-iteration protocol;
and a two-phase algorithm incorporating a {\em trivial barrier} before any
collective communication call.

From the viewpoint of an MPI application, any call to an MPI
collective communication function is interposed on by a wrapper function,
as shown in Algorithm~\ref{algo:wrapper}.

\begin{algorithm}
  \begin{algorithmic}[1]
    \caption{Two-Phase collective communication wrapper.\\
        (This wrapper function interposes on all MPI collective communication
        functions invoked by an MPI application)}
    \label{algo:wrapper}
    \Function{Collective Communication Wrapper}{}
      \State{\Comment{Begin Phase~1}}
      \State{Call MPI\_Barrier()  \Comment{trivial barrier}}
      \State{\Comment{Begin Phase~2}}
      \State{Call original MPI collective communication function}
    \EndFunction
  \end{algorithmic}
\end{algorithm}

Recall that a {\em trivial barrier} is an extra call to {\tt
MPI\_Barrier()} prior to a collective call. A collective MPI call can intuitively be
divided into two parts:  the participating MPI ranks ``register'' themselves
as ready for the collective communication; and then the ``work'' of
communication is carried out.  Where the time for the collective communication
calls of an MPI program is significant, it is typically due to significant
``work'' in the second part of the calls.  Adding a trivial barrier
requires the MPI ranks to register themselves once for the trvial barrier
(but no work is involved), and then register themselves again for the
actual MPI collective communication.  The overhead due to registering
twice is tiny in practice.  Evidence for this can be seen in the
experiments in Section~\ref{sec:micro}, which show small overhead.

The purpose of Algorithm~\ref{algo:wrapper} is to enforce the following
extension of the invariant presented in
Section~\ref{sec:mpiCollectivesOverview}:
\begin{quote}
  {\em No checkpoint must take place while a rank is inside the
  collective communication call (Phase~2) of a wrapper function
  for collective communication (Algorithm~\ref{algo:wrapper}}).
\end{quote}
We formalize this with the following theorem, which guarantees
Algorithm~\ref{algo:mpi2pcSimple} satisfies this invariant.

\begin{theorem}
\label{thm:mpi2pcSimple}
Under Algorithm~\ref{algo:mpi2pcSimple}, an MPI rank is never inside
a collective communication call when a checkpoint message is received
from the checkpoint coordinator.
\end{theorem}

The proof of this theorem is deferred until the end of this subsection.
We begin the path to this proof by stating an axiom that serves
to define the concept of a barrier.

\begin{axiom}
\label{fundamentalAxiom}
For a given invocation of an MPI barrier, it never happens that
a rank~A exits from the barrier before another
rank~B enters the barrier under the ``happens-before'' relation.
\end{axiom}

Next, we present the following two lemmas.

\begin{figure}[t!]
  \centering
  \includegraphics[width=0.8\columnwidth]{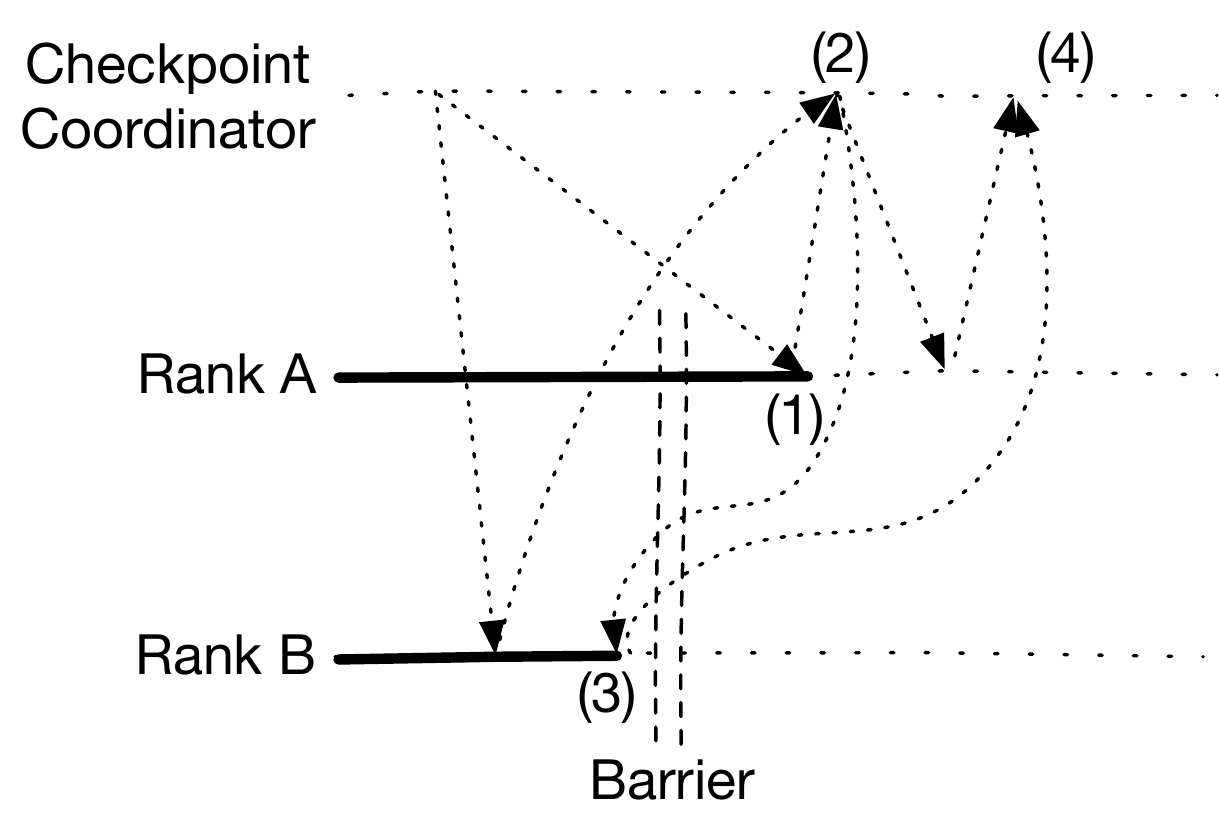}
  \caption{Fundamental ``happens-before'' relation in
    communication between the checkpoint coordinator and the MPI
    ranks involved in an MPI barrier.
    \label{fig:barrier}}
\end{figure}

\begin{lemma}
\label{lemma:mpiBarrier}
For a given MPI barrier, if the checkpoint coordinator sends a message to
each MPI rank participating in the barrier, and if at least one of the
reply messages from the participating ranks reports that its rank has
exited the barrier, then the MPI coordinator can send a second message
to each participating rank, and each MPI rank will reply that it has
entered the barrier (and perhaps also exited the barrier).
\end{lemma}

\begin{proof}
We prove the lemma by contradiction.  Suppose that the lemma does not
hold.  Figure~\ref{fig:barrier} shows the general case in which
this happens.  At event~4, the checkpoint coordinator will conclude that
event~1 (rank~A has exited the MPI barrier) happened before event~2
(the first reply by each rank), which happened before event~3 (in
which rank~B has not yet entered the barrier).  But this contradicts
Axiom~\ref{fundamentalAxiom}.  Therefore, our assumption is false,
and the lemma does indeed hold.
\end{proof}

\begin{lemma}
\label{lemma:wrapperProperties}
Recall that an MPI collective communication wrapper makes a call
to a trivial barrier and then makes an MPI collective communication call.
For a given invocation of an
MPI collective communication wrapper, we know that one of four
cases must hold:
\begin{enumerate}
\item[(a)] an MPI rank is in the collective communication call,
  and all other ranks are either in the call, or have exited;
\item[(b)] an MPI rank is in the collective communication call,
  and no rank has exited, and every other rank has at least entered
  the trivial barrier (and possibly proceeded further);
\item[(c)] an MPI rank is in the trivial barrier and no other rank
  has exited (but some may not yet have entered the trivial barrier);
\item[(d)] either no MPI rank has entered the trivial barrier, or
  all MPI ranks have exited the MPI collective communication call.
\end{enumerate}
\end{lemma}

\begin{proof}
The proof is by repeated application of Lemma~\ref{lemma:mpiBarrier}.
For case~a, if an MPI rank is in the collective
communication call and another rank has exited the collective call,
then Lemma~\ref{lemma:mpiBarrier} says that there cannot be any rank
that has not yet entered the collective call.
For case~b, note that if an MPI rank is in the collective communication
call, then that rank has exited the trivial barrier.  Therefore, by
Lemma~\ref{lemma:mpiBarrier}, all other ranks have at least entered the
trivial barrier.  Further, we can assume that no ranks that
have exited the collective call, since we would otherwise be in case~a,
which is already accounted for.
For case~c, note that if an MPI rank is in the trivial barrier and no
rank has exited the trivial barrier, then Lemma~\ref{lemma:mpiBarrier}
says that there cannot be any rank that has not yet entered the trivial
barrier.
Finally, if we are not in case~a, b, or~c, then the only remaining
possibility is case~d:  all ranks have not yet entered the trivial
barrier or all ranks have exited the collective call.
\end{proof}

\begin{algorithm}
  \begin{algorithmic}[1]
    \caption{\label{algo:mpi2pcSimple} Two-Phase algorithm for deadlock-free
                        checkpointing of MPI collectives}
    \State{{\bf Messages:} \{intend-to-checkpoint, extra-iteration, do-ckpt\}}
    \State{{\bf MPI states:} \{ready, in-phase-1, exit-phase-2\}}
    \medskip

    \Process{Checkpoint Coordinator}
      \Function{Begin Checkpoint}{}
       \State{send intend-to-ckpt msg to all ranks}
       \State{receive responses from each rank}
       \While{some rank in state exit-phase-2}
         \State{send extra-iteration msg to all ranks}
         \State{receive responses from each rank}
        \EndWhile
       \State{send do-ckpt msg to all ranks}
      \EndFunction
    \EndProcess
    \medskip

    \Process{MPI Rank}
      \Event{intend-to-ckpt msg or extra-iteration msg}
        \If{not inCollectiveWrapper}
          \State{reply to ckpt coord:~~ state $\gets$ ready}
        \EndIf
        \If{inCollectiveWrapper and in Phase 1}
          \State{reply to ckpt coord:~~ state $\gets$ in-phase-1}
        \EndIf
        \If{inCollectiveWrapper and in Phase 2}
          \State{\Comment{guaranteed ckpt coord won't request ckpt here}}
          \State{finish executing coll. comm. call}
          \State{reply to ckpt coord:~~ state $\gets$ exit-phase-2}
          \State{\Comment{ckpt coord can request ckpt after this}}
          \State{set state $\gets$ ready}
        \EndIf
        \State{continue, but {\em wait} before next coll. comm. call}
      \EndEvent
      \Event{do-ckpt msg}
        \State{\Comment{guaranteed now that no rank is in phase~2 during ckpt}}
        \State{do local checkpoint for this rank}
        \State{\Comment{all ranks may now continue executing}}
        \If{this rank is waiting before coll. comm. call}
          \State{unblock this rank and continue executing}
        \EndIf
      \EndEvent
    \EndProcess
  \end{algorithmic}
\end{algorithm}

We now continue with the proof of the main theorem
(Theorem~\ref{thm:mpi2pcSimple}), which was deferred earlier.

\begin{proof}
{\em (Proof of Theorem~\ref{thm:mpi2pcSimple}
     for Algorithm~\ref{algo:mpi2pcSimple}.)}
Lemma~\ref{lemma:wrapperProperties} states that one of four cases must
hold in a call by \sol{} to an MPI collective communication wrapper.
We wish to exclude the possibility that an MPI rank is in the collective
communication call (case~a or~b of the lemma) when the checkpoint
coordinator invokes a checkpoint.

In Algorithm~\ref{algo:mpi2pcSimple},
assume that the checkpoint coordinator has sent an {\em intend-to-ckpt}
message, and has not yet sent a {\em do-ckpt} message.
An MPI rank will either reply
with state {\em ready} or {\em in-phase-1} (showing that it is not
in the collective communication call and that it will stop before
entering the collective communication call), or else it must be
in Phase~2 of the wrapper
(potentially within the collective communication call), and it will not
reply to the coordinator until exiting the collective call.
\end{proof}

\begin{theorem}
\label{thm:efficiency}
Under Algorithm~\ref{algo:mpi2pcSimple}, deadlock will never occur.
Further,
the delay between the time when all ranks have received the
{\em intend-to-checkpoint} message and the time when the {\em do-ckpt}
message has been sent is bounded by the maximum time for any
individual MPI rank to enter and exit the collective communication call,
plus network message latency.
\end{theorem}

\begin{proof}
The algorithm will never deadlock, since each rank must either make
progress based on the normal MPI operation
or else it stops {\em before} the collective communication call.
If any rank replies with the state {\em exit-phase-2}, then the checkpoint
coordinator will send an additional {\em extra-iteration} message.
So, at the time of checkpoint, all ranks will have state {\em ready}
or {\em in-phase-1}.

Next, the delay between the time when all ranks have received the
{\em intend-to-checkpoint} message and the time when the {\em do-ckpt}
message has been sent is clearly bounded by the maximum time for an
individual MPI rank to enter and exit the collective communication call,
plus the usual network message latency.
This is the case since once the {\em intend-to-checkpoint} message is
received, no MPI rank may enter the collective communication
call.  So, upon receiving the {\em intend-to-checkpoint} message,
either the rank is already in Phase~2 or else it will remain
in Phase~1.
\end{proof}

\smallskip\noindent{\bf Implementation of Algorithm~\ref{algo:mpi2pcSimple}:}
At the time of process launch for an MPI rank, a separate checkpoint
helper thread is also injected into each rank. This thread is responsible
for listening to checkpoint-related messages from a separate
coordinator process and then responding.  This allows the MPI rank
to asynchronously process events based on messages received from the
checkpoint coordinator.  Furthermore at the time of checkpoint, the
existing threads of the MPI rank process are quiesced (paused) by the
helper thread, and the helper thread carries out the checkpointing
requirements, such as copying the upper-half memory regions to stable
storage.  The coordinator process does not participate in the
checkpointing directly.  In the implementation, a DMTCP coordinator
and DMTCP checkpoint thread~\cite{ansel2009dmtcp} are modified to serve
as checkpoint coordinator and helper thread, respectively.

\subsection{Verification with TLA+/PlusCal}
\label{sec:implementation}

To gain further confidence in our implementation for handling
collective communication (Section~\ref{sec:mpiCollectivesAlgo}), we
developed a model for the protocol in TLA+~\cite{lamport1999tla} and then
used the PlusCal model checker of TLA+ based on TLC~\cite{yu1999model}
to verify Algorithm~\ref{algo:mpi2pcSimple}.  Specifically, PlusCal was
used to verify the algorithm invariants of deadlock-free execution and
consistent state when multiple concurrent MPI processes are executing. The
PlusCal model checker did not report any deadlocks or broken invariants
for our implementation.

\subsection{Checkpoint/Restart Package}
\label{sec:crpackage}
Any single-process checkpointing package could be utilized for the basis of
implementing \sol{}.  This work presents a prototype implemented by
extending DMTCP~\cite{ansel2009dmtcp} and by developing a DMTCP
plugin~\cite{arya2016design}.  Cao et al.~\cite{cao2016system} demonstrated that
DMTCP can checkpoint MPI-based HPCG over 32,752 CPU cores (38~TB) in 11~minutes,
and MPI-based NAMD over 16,368 cores (10~TB) in 2.6~minutes.

DMTCP uses a helper thread inside each application process,
and a coordinated checkpointing protocol by using a
centralized coordinator daemon. Since this was close to the design requirements
of \sol{}, we leveraged this infrastructure and extended the DMTCP coordinator
to implement the two-phase algorithm.

The same approach could be extended to base \sol{} on top of
a different underlying transparent checkpointing package.  For example,
one could equally well have modified an existing MPI coordinator process
to communicate with a custom helper thread in each MPI rank that then
invokes the BLCR checkpointing package when it is required to execute the
checkpoint.  In particular, all sockets and other network communication
objects are inside the lower half, and so even a single-process or
single-host checkpointing package such as BLCR would suffice for this
work.

\section{Experimental Evaluation}
\label{sec:evaluation}

This section seeks to answer the following questions:

\noindent \textbf{Q1:} What is the runtime overhead of running MPI
applications under \sol{}?

\noindent \textbf{Q2:} What are the checkpoint and restart overheads
of transparent checkpointing of MPI applications under \sol{}?

\noindent \textbf{Q3:} Can \sol{} allow transparent switching of MPI
implementations across checkpoint-restart for the purpose of load balancing?

\subsection{Setup}
\label{sec:expSetup}

We first describe the hardware and software setup for \sol{}'s
evaluation.

\subsubsection{Hardware}
\label{sec:hardware}

The experiments were run on the Cori supercomputer~\cite{cori} at
the National Energy Research Scientific Computing Center (NERSC).
As of this writing, Cori is the \#12 supercomputer in the Top-500
list~\cite{top500}. All experiments used the Intel Haswell
nodes (dual socket with a 16-core Xeon E5-2698~v3 each)
connected via Cray's Aries interconnect
network.  Checkpoints were saved to the backend Lustre filesystem.

\subsubsection{Software}
\label{sec:software}

Cori provides modules for two implementations of MPI: Intel MPI
and Cray MPICH. The Cray compiler (based on an Intel compiler) and
Cray MPICH are the recommended way to use MPI, presumably for reasons
of performance.  Cray MPICH version 3.0 was used for the experiments.

\subsubsection{Application Benchmarks}
\label{sec:benchmarks}

\sol{} was tested with five real-world
HPC applications from different computational science domains:

\begin{enumerate}

  \item GROMACS~\cite{gromacs}: Versatile package for
  molecular dynamics, often used for biochemical molecules.

  \item CLAMR~\cite{clamr, clamrSrc}: Mini-application for
  CelL-based Adaptive Mesh Refinement.


  \item miniFE~\cite{herouxminife}: Proxy application for
        unstructured implicit finite element codes.

  \item LULESH~\cite{lulesh}: Unstructured Lagrangian Explicit Shock
  Hydrodynamics

  \item HPCG~\cite{dongarra2016new} (High Performance Conjugate Gradient):
        Uses a variety of linear algebra operations to match a broad set
        of important HPC applications, and used for ranking HPC systems.
\end{enumerate}

\subsection{Runtime Overhead}
\label{sec:runtimeOvhd}

\subsubsection{Real-world HPC Applications}

Next, we evaluate the performance of \sol{} for real-world
HPC applications.  It will be shown that the runtime overhead is close to
\SI{0}{\percent} for miniFE and HPCG, and as much as \SI{2}{\percent} for
the other three applications.  The higher overhead has been tracked down
to an inefficiency in the Linux kernel~\cite{fsgsbaseLWN} in the case of
many point-to-point MPI calls (send/receive) with messages of small size.
This worst case is analyzed further in Section~\ref{sec:patchedKernel},
where tests with an optimized Linux kernel show a worst case runtime
overhead of \SI{0.6}{\percent}.  The optimized Linux kernel is
based on a patch under review for a future Linux version.

\begin{figure}[t!]
  \centering
  \includegraphics[scale=0.5]{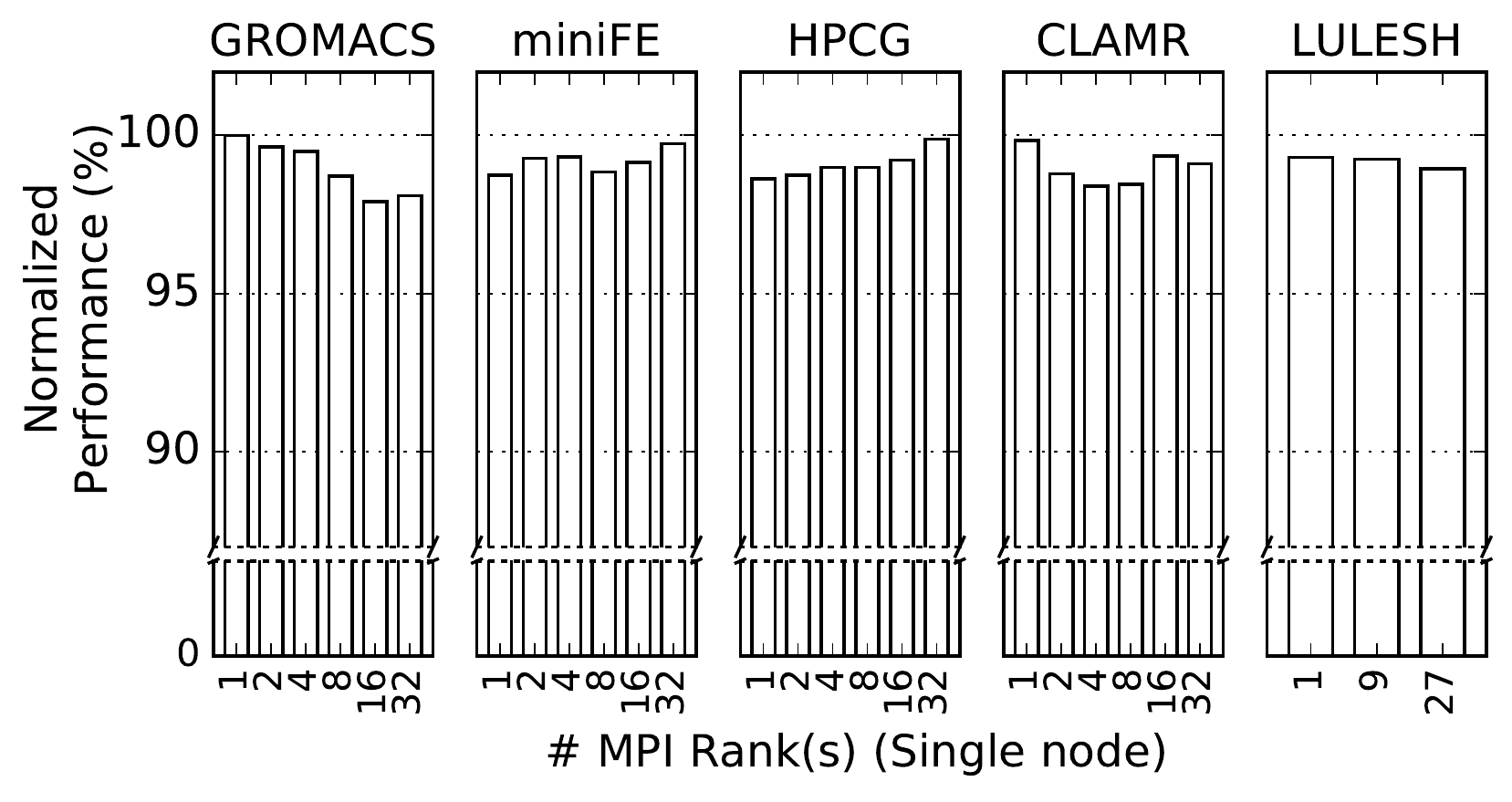}
  \vspace{-0.7cm}
  \caption{Single Node: Runtime overhead under \sol{} for different real-world
  HPC benchmarks with an unpatched Linux kernel. (Higher is better.)
  \label{fig:singleNode}}
  \vspace{-0.4cm}
\end{figure}

\begin{figure}[t!]
  \centering
  \includegraphics[scale=0.5]{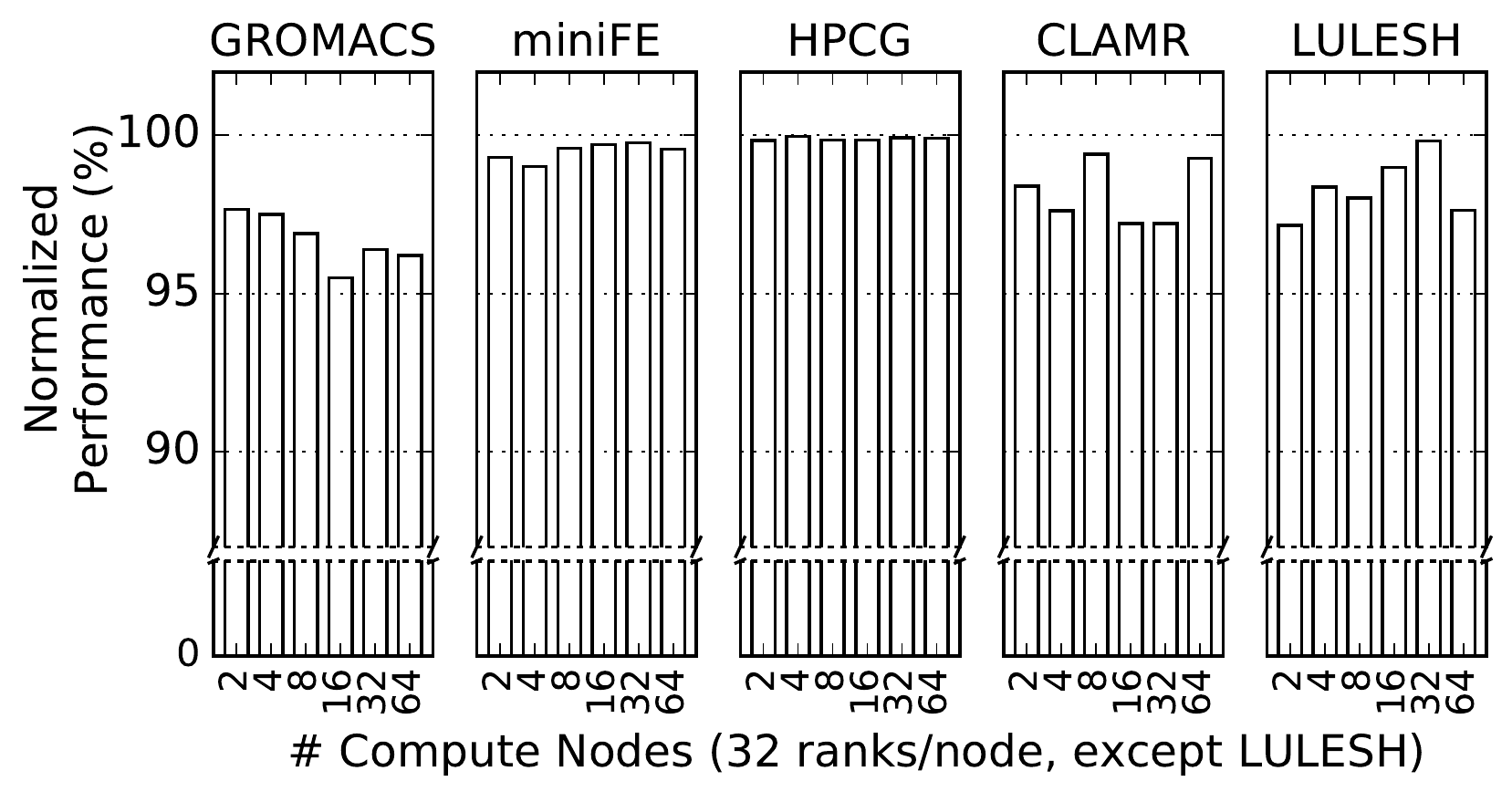}
  \vspace{-0.7cm}
  \caption{Multiple Nodes: Runtime overhead under \sol{} for different
  real-world HPC benchmarks with an unpatched Linux kernel.
  In all cases, except LULESH, 32~MPI ranks
  were executed on each compute node. (Higher is better.)
  \label{fig:multiNode}}
  \vspace{-0.4cm}
\end{figure}

\noindent \textbf{Single Node:} Since the tests were performed within
a larger cluster where the network use of other jobs could create
congestion, we first eliminate any network-related
overhead by running the benchmarks on a single node with
multiple MPI ranks, both under \sol{} and natively (without \sol{}). This
experiment isolates the single-node runtime overhead of \sol{}
by ensuring that all communication among ranks is intra-node.

Figure~\ref{fig:singleNode} shows the results for the five different
real-world HPC applications running on a single node under \sol{}. Each run
was repeated 5~times (including the native runs), and the figure shows the
mean of the 5~runs. The absolute runtimes varied from \SI{4.5}{\minute} to
\SI{15}{\minute}, depending on the configuration. The worst case overhead
incurred by \sol{} is \SI{2.1}{\percent} in the case of GROMACS (with 16~MPI
ranks).
In most cases, the mean overhead is less than \SI{2}{\percent}.

\begin{figure}
  \centering
    \includegraphics[scale=0.3]{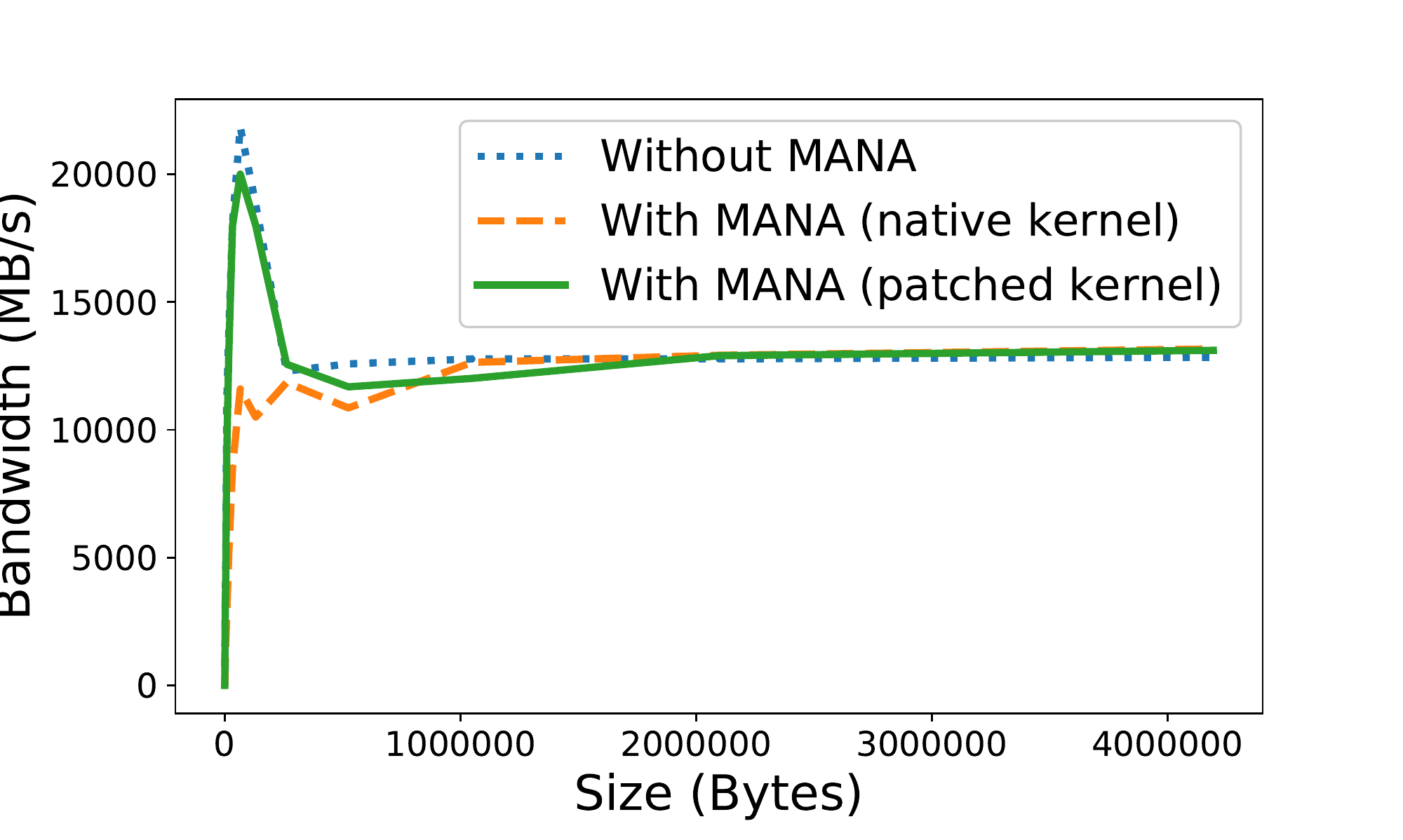}
  \caption{Point-to-Point Bandwidth under \sol{} with patched
  and unpatched Linux kernel.  (Higher is better.)
  \vspace{-0.4cm}
	\label{fig:osuP2Pall}
  }
\end{figure}

\noindent \textbf{Multiple Nodes:} Next, the scaling of \sol{}
across the network is examined for up to 64~compute nodes and with 32~ranks
per node (except for LULESH, whose configuration restricts the number
of ranks/node based on the number of nodes).
Hence, the number of MPI ranks
ranges from 64 to~\num{2048}.

Figure~\ref{fig:multiNode} shows the results of five different real-world
HPC applications running on multiple nodes under \sol{}. Each run was repeated
5 times, and the mean of 5~runs is reported. We observe a trend similar to
the single node case. \sol{} imposes an overhead of typically less than \SI{2}{\percent}.
The highest overhead observed is \SI{4.5}{\percent} in the case of GROMACS (512
ranks running over 16 nodes). However, see Section~\ref{sec:patchedKernel}
where we demonstrate a reduced overhead of \SI{0.6}{\percent} with GROMACS.

\subsubsection{Memory Overhead}

The upper-half libraries were built with \texttt{mpicc}, and hence include
additional copies of the MPI library that are not used.  However, the
upper-half MPI library is never initialized, and so no network
library is ever loaded into the upper half.

Since a significant portion of the lower half is comprised only of the
MPI library and its dependencies, the additional copy of the libraries
(with one copy residing in the upper half) imposes a constant memory
overhead. This text segment (code region) was \SI{26}{\mega\byte} in
all of our experiments on Cori with the Cray MPI library.

In addition to the code, the libraries (for example, the networking
driver library) in the lower half also allocate additional memory
regions (shared memory regions, pinned memory regions, memory-mapped
driver regions). We observed that the shared memory regions mapped by the
network driver library grow in proportion with the number of nodes (up
to 64~nodes): from 2~MB (for 2~nodes) to 40~MB for (64~nodes). We expect
\sol{} to have a reduced checkpoint time compared to
DMTCP/InfiniBand~\cite{cao2014transparent}, as \sol{} discards these
regions during checkpointing, reducing the amount of data that's written
out to the disk.

\subsubsection{Microbenchmarks}
\label{sec:micro}

\begin{figure*}[t!]
  \begin{subfigure}[t]{0.33\textwidth}
    \includegraphics[scale=0.30]{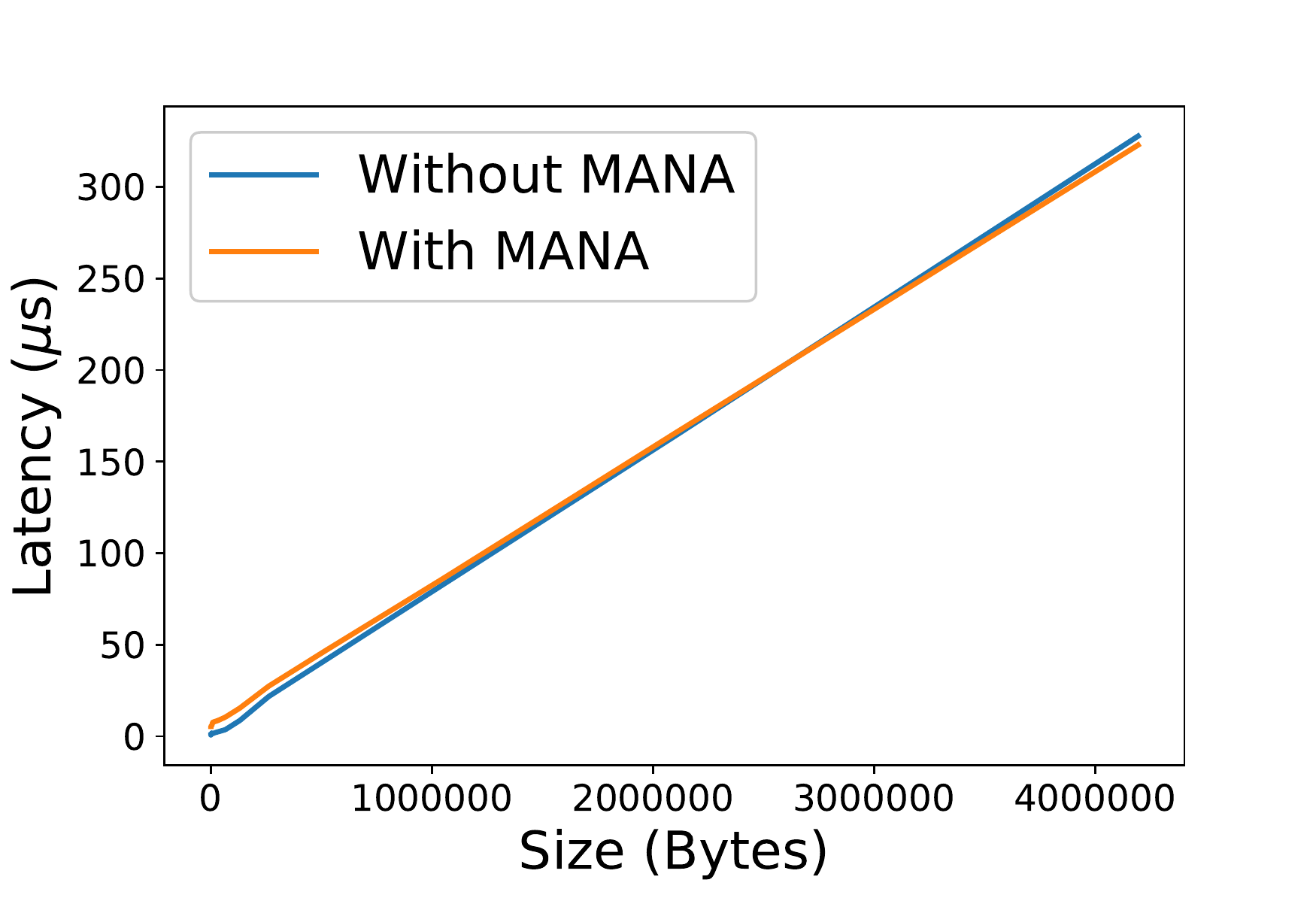}
    \caption{Point-to-Point Latency}
  \end{subfigure}\hfill%
  \begin{subfigure}[t]{0.33\textwidth}
    \includegraphics[scale=0.30]{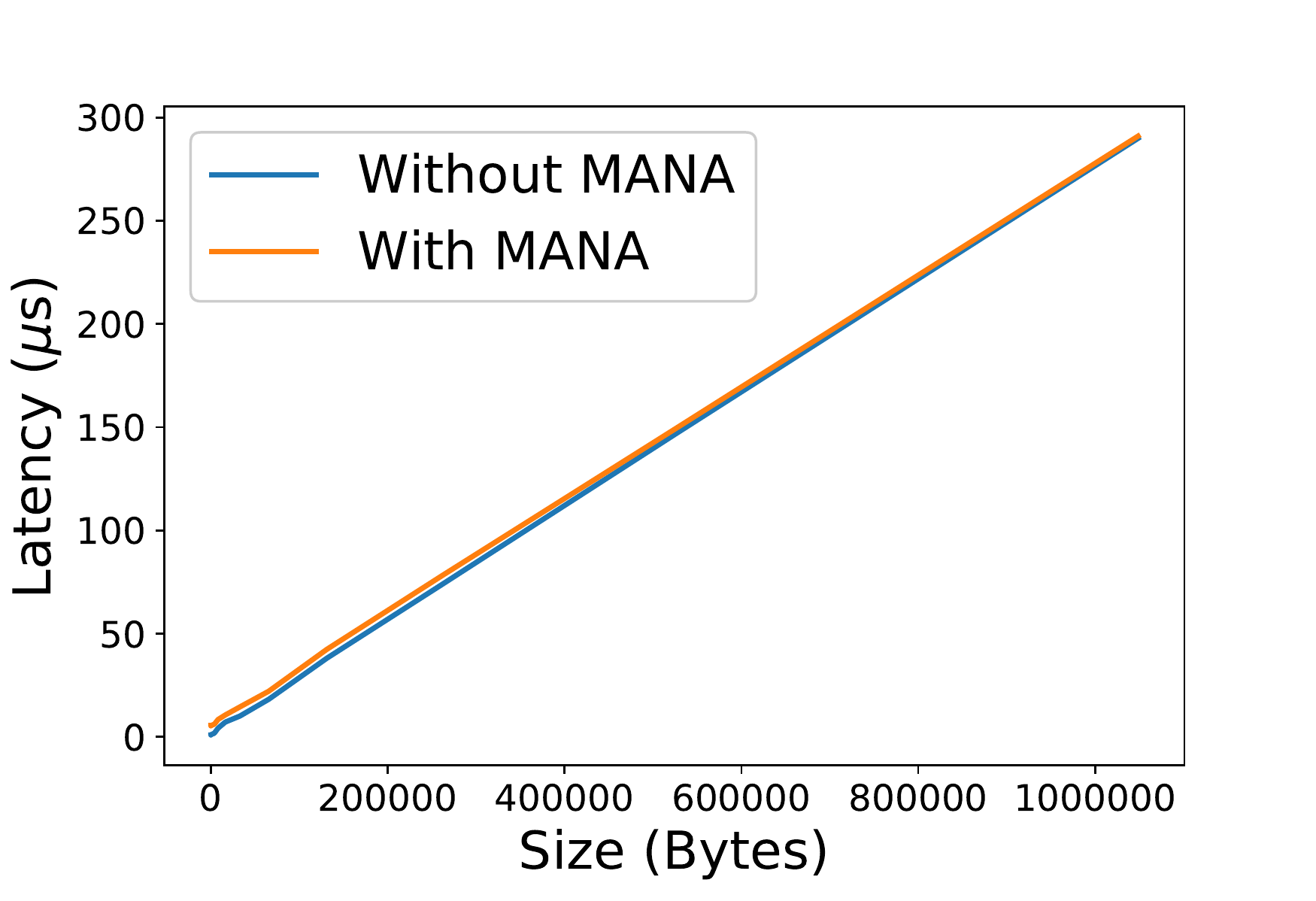}
    \caption{Collective MPI\_Gather}
  \end{subfigure}\hfill%
  \begin{subfigure}[t]{0.33\textwidth}
    \includegraphics[scale=0.30]{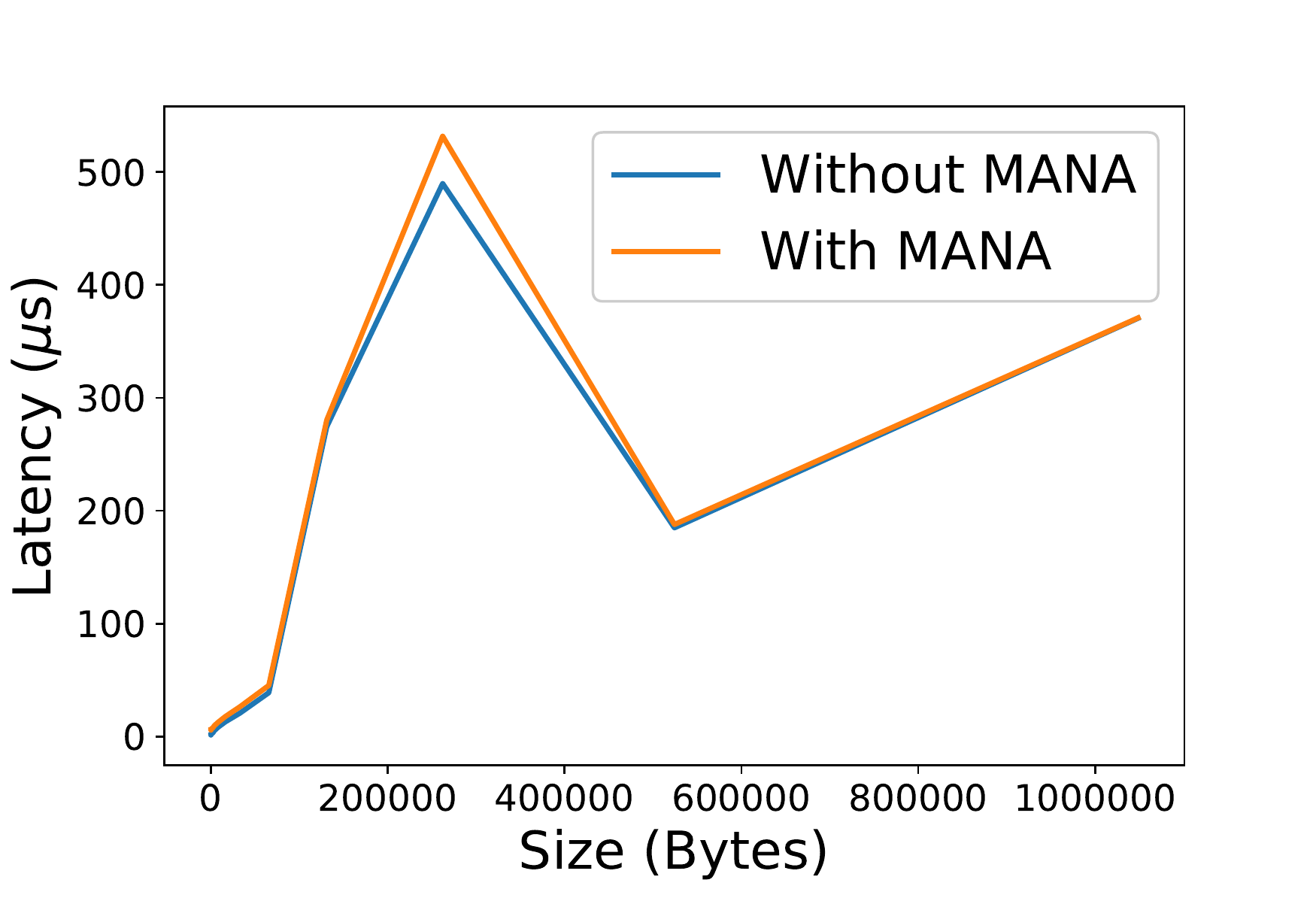}
    \caption{Collective MPI\_Allreduce}
  \end{subfigure}
  \caption{OSU Micro-benchmarks under \sol{}.
  (Results are for two MPI ranks on a single node.)\label{fig:osuResults}}
\end{figure*}

To dig deeper into the sources for the runtime overhead, we tested \sol{}
with the OSU micro-benchmarks. The benchmarks stress and evaluate the bandwidth
and latency of different specific MPI subsystems. Our choice of the specific
micro-benchmarks was motivated by the MPI calls commonly used by our real-world
MPI applications.

Figure~\ref{fig:osuResults} shows the results with three benchmarks from
the OSU micro-benchmark suite. These benchmarks correspond with
the most frequently used MPI subsystems in the set of real-world HPC
applications. The benchmarks were run with 2 MPI ranks running on a single
compute node.

The results show that latency does not suffer under \sol{}, for both
point-to-point and collective communication. (The latency curves for application
running under \sol{} closely follow the curves when the application is run
natively.)

\subsection{Source of Overhead and Improved Overhead for Patched Linux Kernel}
\label{sec:patchedKernel}

All experiments in this section were performed on a single node of
our local cluster, where it was possible to directly install a
patched Linux kernel in the bare machine.

Further investigation revealed two sources of runtime overhead.
The larger source of overhead is due to the use of the ``FS'' register
during transfer of flow of control between the upper and lower half and
back during a call to the MPI library in the lower half.  The ``FS''
register of the x86-64 CPU is used by most compilers to refer to the
thread-local variables declared in the source code.  The upper and lower
half programs each have their own thread-local storage region.  Hence,
when switching between the upper and lower half programs, the value of
the ``FS'' register must be changed to point to the correct thread-local
region.  Most Linux kernels today require a kernel call to invoke
a privileged assembly instruction to get or set the ``FS'' register.
In~2011, Intel Ivy Bridge CPUs introduced a new, unprivileged FSGSBASE
assembly instruction for modifying the ``FS'' register, and a patch
to the Linux kernel~\cite{fsgsbaseLWN} is under review to allow other
Linux programs to use this more efficient mechanism for managing the
``FS'' register. (Other architectures, such as ARM, use unprivileged
addressing modes for thread-local variables that do not depend on
special constructs, such as the x86 segments.)

A second (albeit smaller) source of overhead is the virtualization
of MPI communicators and datatypes, and recording of metadata for MPI
sends and receives. Virtualization requires a hash table lookup and locks
for thread safety.

The first and larger source of overhead is then eliminated by using
the patched Linux kernel, as discussed above.
Point-to-point bandwidth benchmarks were run both with and without the patched
Linux kernel (Figure~\ref{fig:osuP2Pall}).
A degradation in runtime performance is seen for \sol{} for small message
sizes (less than \SI{1}{\mega\byte}) in the case of a native kernel.
However, the figure shows that the patched kernel yields much
reduced runtime overhead for \sol{}.
Note that the Linux kernel community is actively reviewing this
patch (currently in its third version), and it is likely to be
incorporated in future Linux releases.

Finally, we return to GROMACS, since it exhibited a higher
runtime overhead (e.g., \SI{2.1}{\percent} in the case of 16~ranks) in
many cases.
We did a similar experiment, running GROMACS with 16~MPI
ranks on a single node with the patched kernel. With the patched kernel,
the performance degradation was reduced to \SI{0.6}{\percent}.

\subsection{Checkpoint-restart Overhead}
\label{sec:ckptRstOvhd}

\begin{figure}[t!]
  \centering
  \includegraphics[scale=0.5]{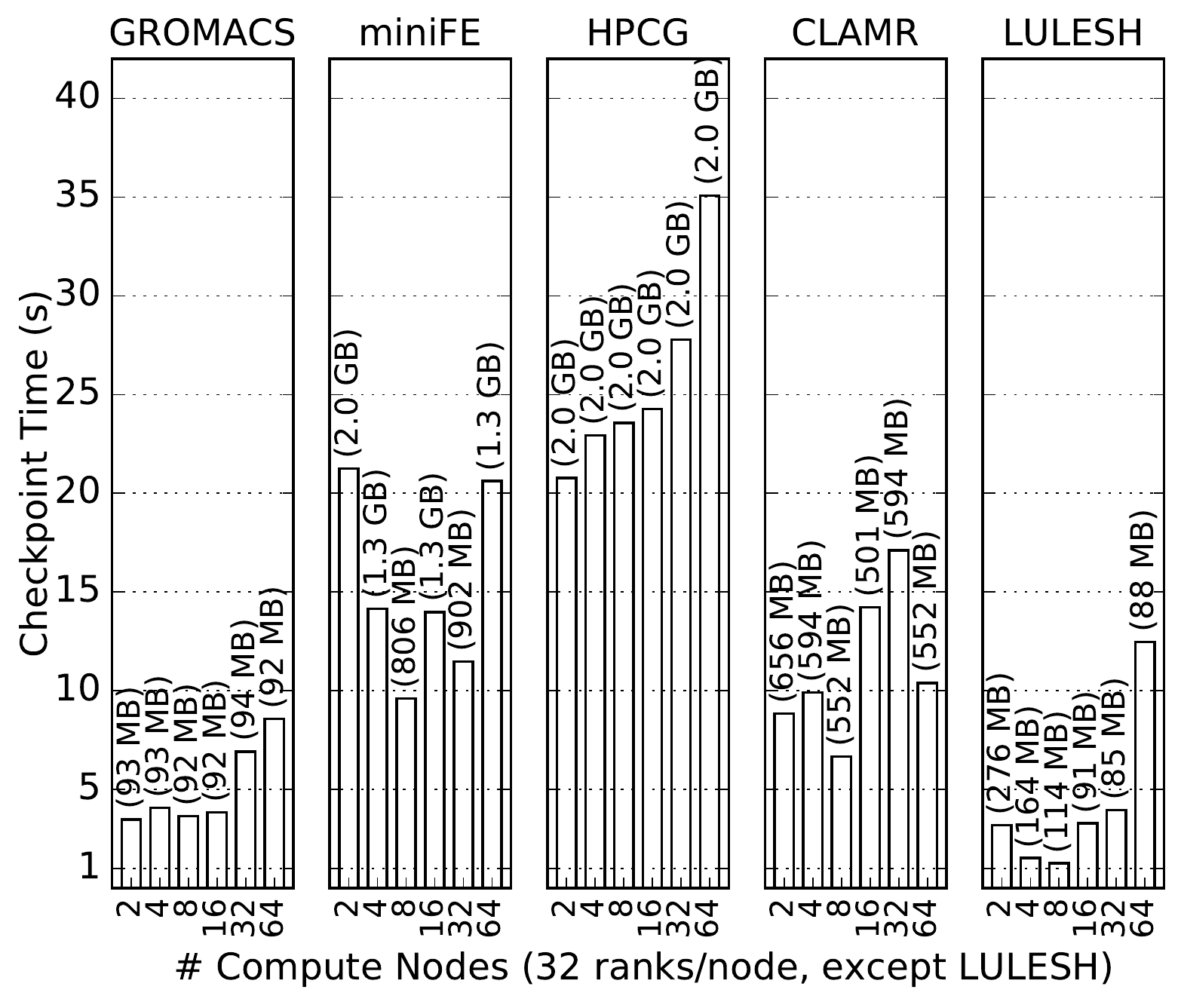}
  \caption{Checkpointing overhead and checkpoint image sizes under \sol{} for
  different real-world HPC benchmarks running on multiple nodes. In all cases,
  except LULESH, 32 MPI ranks were executed on each compute node. For LULESH,
  the total number of ranks was either 64 (for 2, 4, and 8 nodes), or 512 (for
  16, 32, and 64 nodes). Hence, the maximum number of ranks (for 64~nodes)
  was~2048.
  The numbers above the bars (in parentheses) indicate the
  checkpoint image size for each MPI rank.
  \label{fig:ckptTimes}}
  \vspace{-0.5cm}
\end{figure}

\begin{figure}[t!]
  \centering
  \includegraphics[scale=0.5]{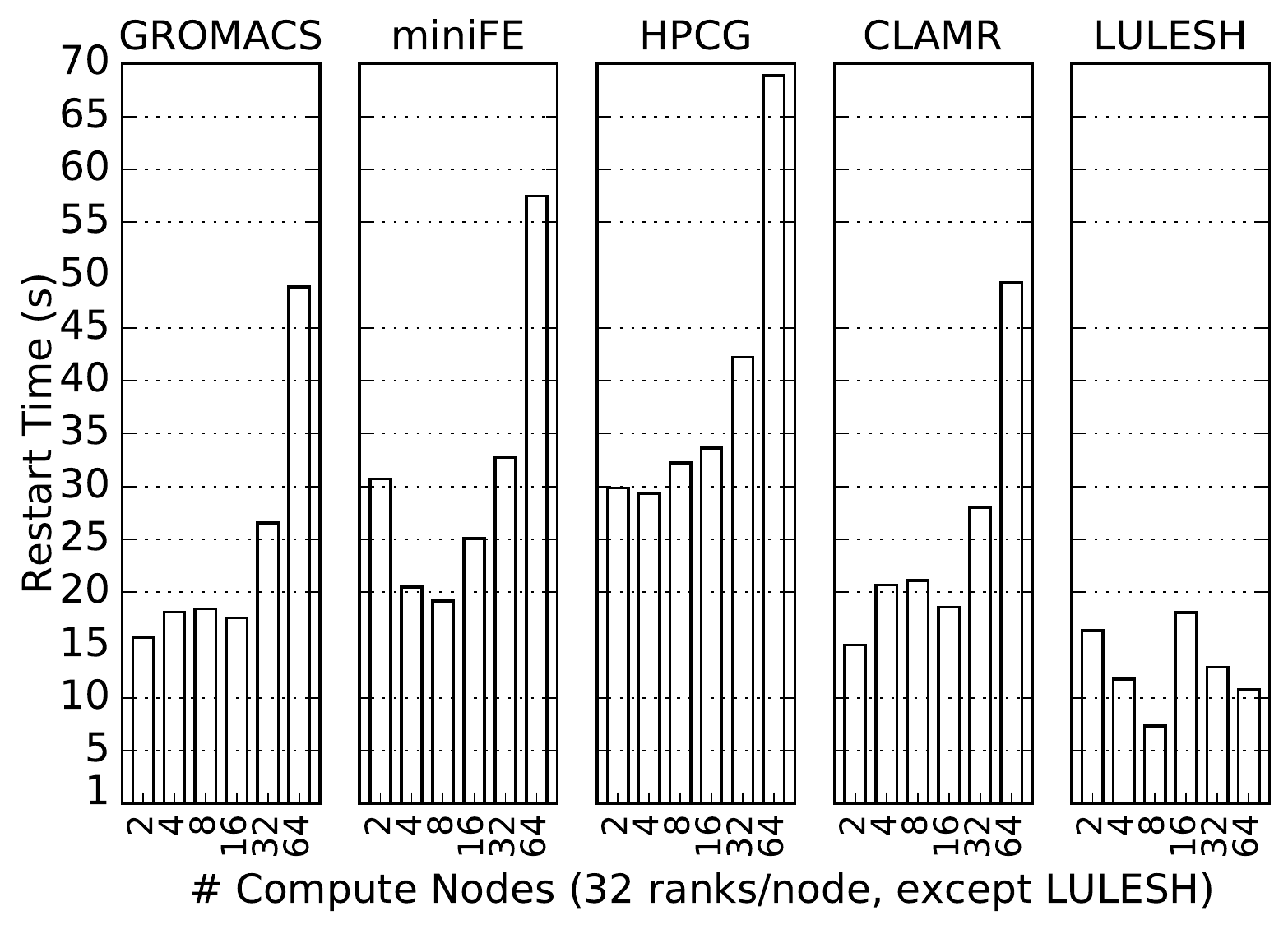}
  \caption{Restart overhead under \sol{} for different real-world HPC
  benchmarks running on multiple nodes. In all cases, except LULESH, 32 MPI
  ranks were executed on each compute node.
  Ranks/node is as in Figure~\ref{fig:ckptTimes}.
  \label{fig:restartTimes}}
  \vspace{-0.5cm}
\end{figure}

In this section, we evaluate \sol{}'s performance when checkpointing and
restarting HPC applications.  Figure~\ref{fig:ckptTimes} shows the
checkpointing overhead for five different real-world HPC applications
running on multiple nodes under \sol{}. Each run was repeated 5 times,
and the mean of five runs is reported. For each run, we use the \texttt{fsync}
system call to ensure the data is flushed to the Lustre backend storage.

The total checkpointing data written at each checkpoint varies from
\SI{5.9}{\giga\byte} (in the case of 64~ranks of GROMACS running over 2~nodes)
to \SI{4}{\tera\byte} (in the case of 2048 ranks of HPCG running over
64~nodes). Note that the checkpointing overhead is proportional
to the total
amount of memory used by the benchmark. This is also reflected in the size
of the checkpoint image per MPI rank. While 
Figure~\ref{fig:ckptTimes} reports the overall checkpoint
time, note that there is significant variation in the
write times for each MPI rank during a given run.
(The time for one rank to write its checkpoint data can be up to
4~times more than that for \SI{90}{\percent} of the other ranks.)
This phenomenon of stragglers
during a parallel write has also been noted by other
researchers~\cite{arya2016design,xie2012characterizing}. Thus, the overall
checkpoint time is bottlenecked by the checkpoint time of the slowest rank.

\begin{figure}[t!]
  \centering
  \includegraphics[scale=0.5]{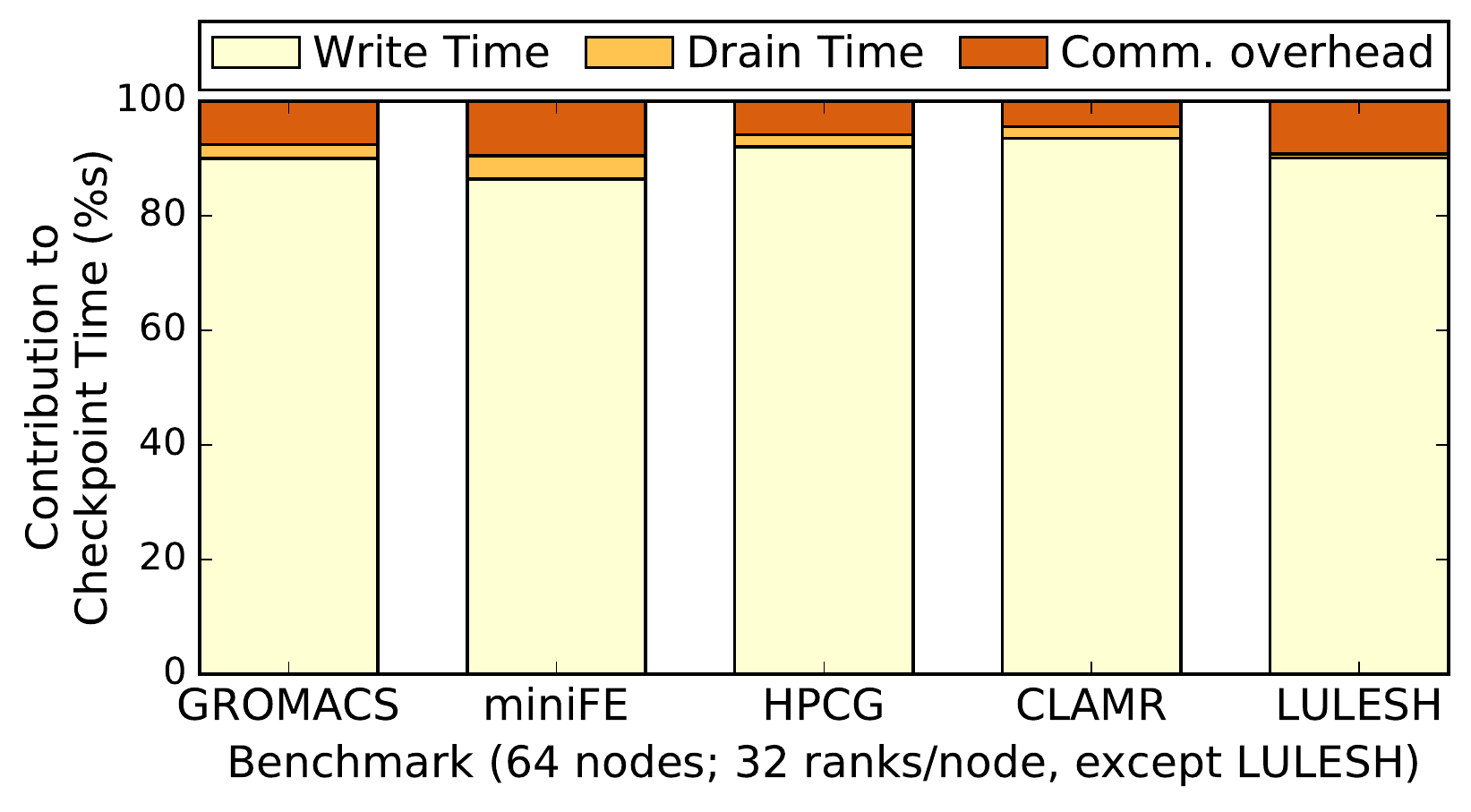}
  \caption{Contribution of different factors to the checkpointing overhead
  under \sol{} for different real-world HPC benchmarks running on 64 nodes.
  Ranks/node is as in Figure~\ref{fig:ckptTimes}.
  The ``drain time'' is the delay in starting a checkpoint
  while MPI message in transit are completed.  The communication overhead
  is the time required in the protocol for network communication between
  the checkpoint coordinator and each rank.
  \label{fig:ckptTimeDivision}}
  \vspace{-0.4cm}
\end{figure}

Next, we ask what are the sources of the checkpointing
overhead? Does the draining of MPI messages and the two-phase algorithm impose
a significant overhead at checkpoint time?

Figure~\ref{fig:ckptTimeDivision} shows the contribution of different
components to the checkpointing overhead for the case of 64~nodes for
the five different benchmarks. In all cases, 
the communication overhead for handling MPI collectives
in the two-phase algorithm of Section~\ref{sec:mpiCollectivesAlgo}
is found to be less than \SI{1.6}{\second}.

In all cases, the time to drain in-flight MPI messages was
less than \SI{0.7}{\second}. The total checkpoint time was
dominated by the time to write to the storage system. The next big source
of checkpointing overhead was the communication overhead. The current
implementation of the checkpointing protocol in DMTCP uses TCP/IP sockets
for communication between the MPI ranks and the centralized DMTCP coordinator.
The communication overhead associated with the TCP layer
is found to increase with the number of ranks,
especially due to metadata in the case of small messages that are exchanged
between MPI ranks and the coordinator.

Finally, Figure~\ref{fig:restartTimes} shows the restart overhead
under \sol{} for the different MPI
benchmarks. The restart time varies from less than \SI{10}{\second}
to \SI{68}{\second} (for 2048 ranks of HPCG running over 64~nodes). The restart
times increase in proportion to the total amount of checkpointing data
that is read from the storage. In all the cases, the restart overhead is
dominated by the time to read the data from the disk. The time to recreate the
MPI opaque identifiers (see Section~\ref{sec:mpiIds}) is less than
\SI{10}{\percent} of the total restart time.

\subsection{Transparent Switching of MPI libraries across Checkpoint-restart}

This section demonstrates that \sol{} can transparently switch between
different MPI implementations across checkpoint-restart. This is
useful for debugging programs (even the MPI library) as it allows a program
to switch from a production version of an MPI library to a debug version of
the MPI library.

The GROMACS application is launched using the production version of
CRAY~MPI, and a checkpoint is taken \SI{55}{\second} into the run.
The computation is then restarted on top of a custom-compiled debug
version of MPICH (for MPICH version~3.3).  MPICH was chosen because it is
a reference implementation whose simplicity makes it easy to instrument
for debugging.

\subsection{Transparent Migration across Clusters}
\label{sec:crossClusterMigration}

\begin{figure}[t!]
  \centering
   \includegraphics[scale=0.35]{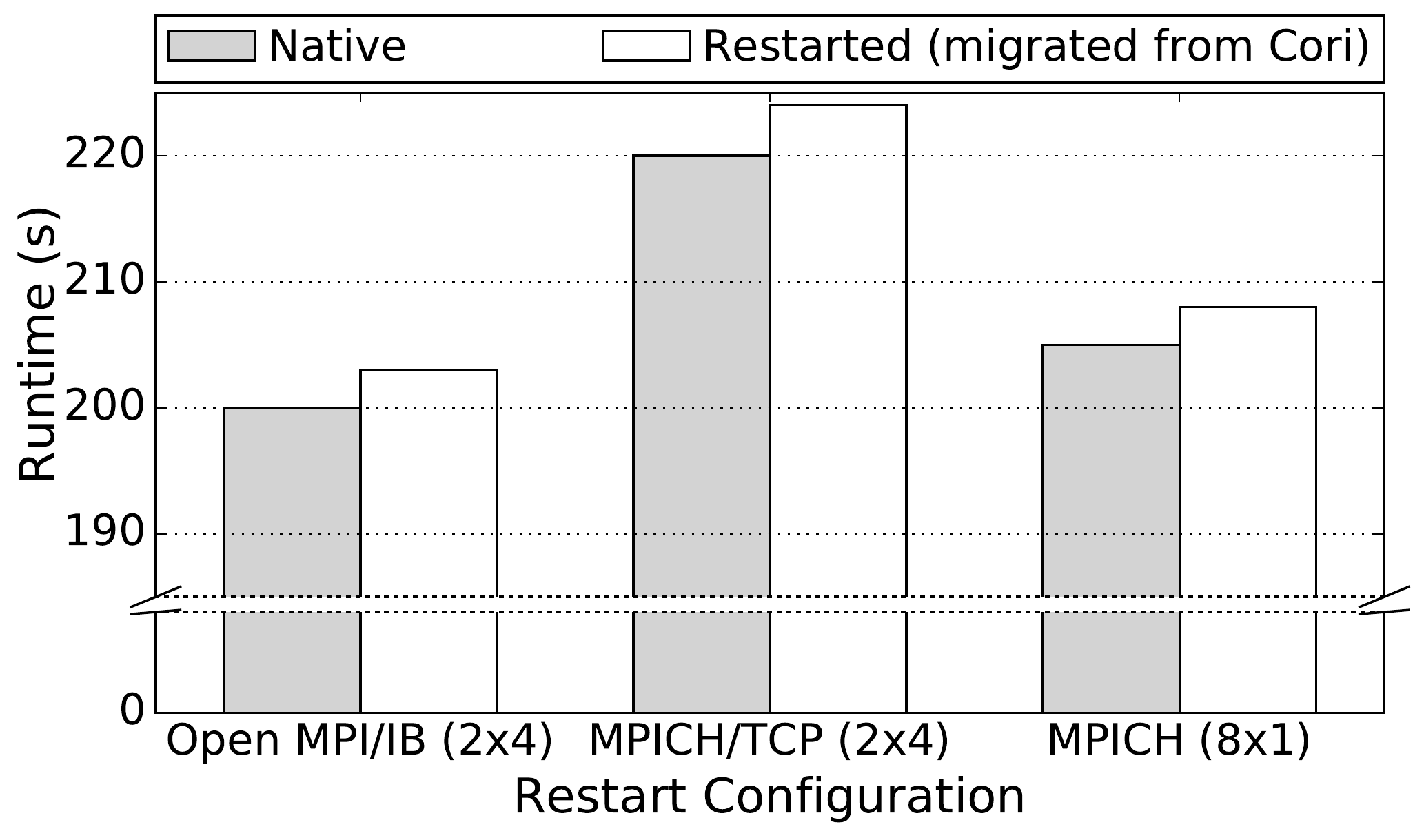}
  \caption{Performance degradation of GROMACS after cross-cluster migration
  under three different restart configurations. The application was restarted
  after being checkpointed at the half-way mark on Cori. (Lower is better.)
	\label{fig:clusterMigration}} 
\end{figure}

Next, we consider cross-cluster migration for purposes of wide-area
load balancing either among clusters at a single HPC site or
even among multiple HPC sites.  This is rarely done, since
the two common vehicles for transparent checkpoint (BLCR as
the base of an MPI-specific checkpoint-restart service; or DMTCP/InfiniBand)
both save the MPI library within the checkpoint image and continue
to use that same MPI library on the remote cluster after migration.
At each site and for each cluster, administrators typically
configure and tune a locally recommended MPI implementation for performance.
Migrating an MPI application {\em along with its underlying MPI library}
destroys the benefit of this local performance tuning.

This experiment showcases the benefits of MPI-agnostic,
network-agnostic support for transparent checkpointing.
GROMACS is run under \sol{},
initially running on Cori with a statically linked Cray~MPI library
running over the Cray Aries network.  GROMACS on Cori is configured to run with
8~ranks over 4~nodes (2~ranks per node).  Each GROMACS rank is
single-threaded.  
A checkpoint was then taken
exactly half way into the run.
The checkpoints were then migrated
to a local cluster that uses Open~MPI over the InfiniBand
network.

The restarted GROMACS under \sol{} was compared with three other
configurations:  GROMACS using the local Open~MPI, configured to use
the local InfiniBand network (8~ranks over 2~nodes); GROMACS/MPICH,
configured to use TCP (8~ranks over 2~nodes); and GROMACS/MPICH, running
on a single node (8~ranks over 1~node). The network-agnostic nature of
\sol{} allowed the Cori version of GROMACS to be restarted on the
local cluster with any of three network options.

We wished to isolate the effects due to \sol{} from the effects due
to different compilers on Cori and the local cluster.  In order to
accomplish this, the native GROMACS on the local cluster was
compiled specially.  The Cray compiler of Cori (using Intel's
C~compiler) was used to generate object files (.o~files) on Cori.
Those object files were copied to the local cluster.  The native
GROMACS was then built using the local \texttt{mpicc}, but with the
(.o~files) as input instead of the (.c~files).  The local \texttt{mpicc}
linked these files with the local MPI implementation, and the
native application was then launched in the traditional way.

Figure~\ref{fig:clusterMigration} shows that GROMACS's performance
degrades by less than 1.8\% post restart on the local
cluster for the three different restart configurations (compared to the
corresponding native runs). Also, note that the performance of GROMACS under
\sol{} post restart closely tracks the performance of the native configuration.

\section{Discussion and Future Work}
\label{sec:futureWork}

Next, we discuss both the limitations and some future implications
of this work concerning dynamic load balancing.

\subsection{Limitations}
While the split-process approach for checkpointing and process migration
is quite flexible, it does include some limitations inherited by any
approach based on transparent checkpointing.  Naturally, when restarting
on a different architecture, the CPU instruction set must be compatible.
In particular, on the x86 architecture, the MPI application code must
be compiled to the oldest x86 sub-architecture among those remote clusters
where one might consider restarting a checkpoint image.  (However, the
MPI libraries themselves may be fully optimized for the local architecture,
since restarting on a remote cluster implies using a new lower half.)

Similarly, while MPI implies a standard API, any local extensions to
MPI must be avoided.  The application {\em binary} interface (ABI)
used by the compiled MPI application must either be compatible or else
a ``shim'' layer of code must be inserted in the wrapper functions for
calling from the upper half to the lower half.

And of course, the use of a checkpoint coordinator implies coordinated
checkpointing.  If a single MPI rank crashes, \sol{} must restore the
entire MPI computation from an earlier checkpoint.

\subsection{Future Work}

MPI version~3 has added nonblocking collective communication calls
(e.g., MPI\_Igather).  In future work, we propose to extend the two-phase
algorithm for collective communication of Section~\ref{sec:mpiCollectivesAlgo} to the
nonblocking case.  The approach to be explored would be to employ a
first phase that uses a nonblocking trivial barrier (MPI\_Ibarrier), and
to then convert the actual asynchronous collective call to a synchronous
collective call (e.g., MPI\_Gather to MPI\_Igather) for the second phase.
Nonblocking variations of collective communication calls are typically
used as performance optimizations in an MPI application.  If an MPI rank
reaches the collective communication early, then instead of blocking,
it can continue with an alternate compute task while occasionally
testing (via MPI\_Test/MPI\_Wait) to see if the other ranks have all
reached the barrier.  In the two-phase analog, a wrapper around the
nonblocking collective communication causes MPI\_Ibarrier to be invoked.
When the ranks have all reached the nonblocking trivial barrier and
the MPI\_Test/MPI\_Wait calls of the original MPI application reports
completion of the MPI\_Ibarrier call of phase~1, then this implies that
the ranks are all ready to enter the actual collective call of phase~2.
A wrapper around MPI\_Test/MPI\_Wait can then invoke the actual collective
call of phase~2.

The split-process approach of \sol{} opens up some
important new features in managing long-running MPI applications.
An immediately obvious feature is the possibility of switching
{\em in the middle of a long run} to a customized MPI implementation.
Hence, one can dynamically substitute a customized MPI for performance
analysis (e.g., using PMPI for profiling or tracing); or using a specially
compiled ``debug'' version of MPI to analyze a particular but occurring
in the MPI library in the middle of a long run.

This work also helps support many tools and proposals for optimizing
MPI applications.  For example, a trace analyzer is sometimes used
to discover communication hotspots and opportunities for better load
balancing.  Such results are then fed back by re-configuring the
binding of
MPI ranks to specific hosts in order to better fit the underlying interconnect topology.

\sol{} can enable new approaches to dynamically load balance across clusters
and also to re-bind MPI ranks in the middle of a long run to create new
configurations of rank-to-host bindings (new topology mappings). Currently,
such bindings are chosen statically and used for the entire lifetime of the
MPI application run. This added flexibility allows system managers to burst
current long-running applications into the Cloud during periods of heavy usage
or when the the MPI application enters a new phase for which a different
rank-to-host binding is optimal.

Finally, \sol{} can enable a new class of very long-running MPI applications ---
ones which may outlive the lifespan of the original MPI Implementation, 
cluster, or even the network interconnect.  Such temporally complex computations
might be discarded as infeasible today without the ability to migrate MPI
implementations or clusters.

\section{Related Work}
\label{sec:relatedWork}

Hursey et al.~\cite{hursey2009interconnect} developed a semi-network-agnostic
checkpoint service for Open-MPI. It applied an ``MPI Message''
abstraction to a Chandy/Lamport algorithm~\cite{ChandyLamport}, greatly
reducing the complexity to support checkpoint/restart for many multiple
network interconnects.  However, it also highlighted the weakness of
implementing transparent checkpointing within the MPI library, since
porting to an additional MPI implementation would likely require
as much software development as for the first MPI implementation.
Additionally, its dependence on BLCR imposed a large overhead cost,
as it lacks support for SysV shared memory.

Separate proxy processes for high- and low-level operations
have been proposed both by CRUM (for CUDA) and McKernel (for
the Linux kernel).
CRUM~\cite{garg2018crum}\ showed that by running a non-reentrant
library in a separate process,
one can work around the problem of a library
``polluting'' the address space of the application process --- i.e.,
creating and leaving side-effects in the application process's
address space. This decomposition of a single application process
into two processes, however, forces the transfer of data between two
processes via RPC, which can cause a large overhead.

McKernel~\cite{gerofi2016scalability} runs a ``lightweight'' kernel along
with a full-fledged Linux kernel. The HPC application runs on the
lightweight kernel, which implements time-critical
system calls. The rest of the functionality is offloaded to a proxy process
running on the Linux kernel. The proxy process is mapped in the address
space of the main application, similar to \sol{}'s concept of
a lower half, to minimize the overhead of ``call forwarding'' (argument
marshalling/un-marshalling).

In general, a proxy process approach is problematic for MPI, since
it can lead to additional jitter as the operating system tries to
schedule the extra proxy process alongside the application
process.  The jitter harms performance since the MPI computation is
constrained to complete no faster than its slowest MPI rank.

Process-in-process~\cite{hori2018process} has in common with
\sol{} that both approaches load multiple programs into a single address
space.  However, the goal of process-in-process was intra-node communication
optimization, and not checkpoint-restart.
Process-in-process loads {\em all} MPI ranks co-located on the same node
as separate threads within a single process, but
in different logical ``namespaces'', in the sense of the \texttt{dlmopen}
namespaces in Linux.
It would be difficult to adapt process-in-process for use in checkpoint-restart
since that approach implies a single ``ld.so'' run-time linker library
that managed all of the MPI ranks.  In particular, difficulties occur
when restarting with fresh MPI libraries while ``ld.so'' retains pointers
to destructor functions in the pre-checkpoint MPI libraries.

In the special regime of application-specific checkpointing
for bulk synchronous MPI applications,
Sultana et al.~\cite{sultana2018mpi} supported checkpointing
by separately saving and restoring MPI state (MPI
identifiers such as communicators, and so on).
This is combined with application-specific code to save the application
state.
Thus, when a live process fails,
it is restored using these two components,
without the need restart the overall MPI job.

SCR~\cite{scrWebsite}, and FTI~\cite{bautista2011fti} are
other application-specific checkpointing techniques.
An application developer declares memory regions they'd like to checkpoint
and checkpointing can only be done at specific points in the program
determined by the application developer. Combining these techniques with
transparent checkpointing is outside the scope of this work, though it is
an interesting avenue for further inquiry.

In general, application-specific and transparent checkpointing each have their
merits.
Both application-specific and transparent checkpointing are used in practice.

At the high end of HPC, application-specific checkpointing is preferred
since the labor for supporting this is small compared to the labor already
invested in supporting an extreme HPC application.

At the low and medium end of HPC, developers prefer transparent checkpointing
because the development effort for the software is more moderate, and the
labor overhead of a specialized application-specific checkpointing solution
would then be significant. System architectures based on burst buffers (e.g., Cray's
DataWarp~\cite{henseler2016datawarp}) can be used to reduce the checkpointing
overhead for both application-specific and transparent checkpointing.

\section{Conclusion}
\label{sec:conclusion}

This work presents an MPI-Agnostic, Network-Agnostic transparent
checkpointing methodology for MPI (\sol{}), based on a {\em split-process}
mechanism.
The runtime overhead is typically less than~2\%,
even in spite of the overhead incurred by
the current Linux kernel when the ``FS'' register is modified each time control
passes between upper and lower half.
Further, Section~\ref{sec:patchedKernel} shows that a
commit (patch) to fix this by the Linux kernel developers is under review
and that this commit reduces the runtime overhead of GROMACS from 2.1\%
to 0.6\% using the patched kernel.  Similar reductions to about
0.6\% runtime overhead are expected in the general case.

An additional major novelty is the demonstration of practical, efficient
migration between clusters at different sites using different networks
and different configurations of CPU cores per node.  This was
considered impractical in the past because a checkpoint image from
one cluster will not be tuned for optimal performance on the second cluster.
Section~\ref{sec:crossClusterMigration} demonstrates that this is now
feasible, and that the migration of
a GROMACS job with 8~MPI ranks experiences an
average runtime overhead of less than 1.8\% as compared
to the native GROMACS application (without \sol{}) on the remote cluster.
As before, even this overhead of 1.8\% is likely to be
reduced to about 0.6\% in the future, based on the results of
Section~\ref{sec:patchedKernel} with a patched Linux kernel.

\section*{Acknowledgment}

We thank Zhengji Zhao and Rebecca-Hartman Baker from NERSC for the
resources and feedback on an earlier version of the software.
We also thank Twinkle Jain for discussions and insights into
an earlier version of this work.

\end{document}